\documentclass{article}
\usepackage{amssymb}
\usepackage{tikz}
\usepackage{amsmath}

\newtheorem{theorem}{Theorem}

\newtheorem{corollary}{Corollary}[theorem]

\newtheorem{property}{Property}

\newcommand{\qed}{\ifhmode\unskip\nobreak\fi\ifmmode\ifinner
\else\hskip5 pt\fi\fi \hbox{\hskip5 pt
\vrule width4 pt  height6 pt  depth1.5 pt \hskip 1pt }}

\title{ Integer Laplacian eigenvalues\\ of  strictly chordal graphs}
\author{Nair Abreu\\
PEP-COPPE - Universidade Federal do Rio de Janeiro\\ 
nairabreunovoa@gmail.com
\and
Claudia Marcela Justel\\
Departamento de Engenharia de Computa\c c\~ao\\ Instituto Militar de Engenharia \\
cjustel@ime.eb.br
\and
Lilian Markenzon\\
PPGI/NCE - Universidade Federal do Rio de Janeiro\\ 
markenzon@nce.ufrj.br
}
\date{\ }

\begin{document}
\maketitle

\begin{abstract}
In this paper, we establish the relation between classic invariants of graphs and their integer Laplacian eigenvalues, 
focusing on a subclass of chordal graphs, the  strictly chordal graphs, and
pointing out how their computation can be efficiently implemented. 
Firstly we review results concerning general graphs showing that 
the number of universal vertices and the degree of false and true twins provide integer Laplacian eigenvalues and their multiplicities. 
Afterwards, we prove that many integer Laplacian eigenvalues of 
a strictly chordal graph are directly related to  particular simplicial vertex sets 
and to the minimal vertex separators of the graph.
 \end{abstract}

Keywords:
integer Laplacian eigenvalue,  strictly chordal graph, universal vertex, false and true  twins,  minimal vertex separator.

%********************************************************************************
%	Introduction								*
%********************************************************************************

\section{Introduction}\label{intro}

Let  $G=(V,E)$  be a connected graph, 
where $|E|=m$ is its {\em size} and
 $|V| = n $ is its {\em order}. 
The {\em neighborhood\/} of a vertex $v \in V$ is denoted by
$N(v) = \{ w \in V; \{v,w\} \in E\}$ and its {\em closed neighborhood} by  $N[v] = N(v)\cup \{v\} $.
Two vertices $u,v \in V$  are  {\em false twins} if  $N(u) = N(v)$
and  {\em true twins} if $N[u] = N[v]$.
For any $S \subseteq V$, 
the subgraph of $G$ induced by $S$ is denoted by $G[S]$ and  
if it is a complete subgraph then $S$ is a \emph{clique} in $G$. 
 The complete graph of order  $n$  is denoted by $K_n$.
A vertex $v$ is said to be {\em
simplicial\/} in $G$ when $N(v)$ is a clique in $G$; it is said to be {\em universal} when $N[v] = V$.

It is important to mention two kinds of cliques in a  chordal graph $G$. 
A  {\em simplicial clique} is a maximal clique containing at least one simplicial vertex.
A simplicial clique $Q$ is called a {\em boundary clique} if there exists a maximal clique $Q^\prime$
such that  $Q \cap Q^\prime $ is the set of non-simplicial vertices of $Q$. 

For $i$, $1 \leq i \leq n$, let $d_i$ be the degree of vertex $v_i$ of $G$.
The Laplacian matrix of $G$ of order $n$ is defined as $L(G) = D(G) - A(G)$, 
where $D(G) = diag(d_1,...,d_n)$ denotes the diagonal degree matrix and $A(G)$ the adjacency matrix of $G$.
As $L(G)$ is symmetric and positive semidefinite all its eigenvalues are non-negative real numbers. 
We denote the eigenvalues of $L(G)$, called the \emph{Laplacian eigenvalues}  of $G$, 
by $\mu_1(G)\geq \dots \geq \mu_n(G)$. 
All different  Laplacian eigenvalues of $G$ together with their multiplicities form the 
Laplacian spectrum  of $G$,  denoted by $Spec L(G)$.  
A graph is called \emph{Laplacian integral}  if its spectrum consists of integers;
in the literature there are several articles about it  \cite{FKMN05,FI95,GMS90, GM94,Ki05,K10,Me94}. 

This paper resumes the subject already treated by the authors in \cite{Ab18} and \cite{Ab19}. 
We establish the relation between classic invariants of graphs and their integer Laplacian eigenvalues, 
pointing out how their computation can be efficiently implemented. 
We focus on a subclass of chordal graphs \cite{Ha63}, the  {\em block duplicate graphs}, introduced by Golumbic and Peled \cite{GP02} and also  defined as {\em strictly chordal graphs} based on hypergraph properties \cite{K05, KLY06};
this class contains the classes of block graphs \cite{Ha63}, block-indifference graphs \cite{Ab18},  
the generalized core-satellite graphs \cite{EB17} and the $(k, t)$-split graphs \cite{Ab19}.
In Section 2, we review results concerning general graphs, showing that 
the number of universal vertices and the degree of false and true twins can provide integer Laplacian eigenvalues and their multiplicities. 
In Section 3, we prove that many integer Laplacian eigenvalues of 
a strictly chordal graph are directly related to   particular simplicial vertex sets 
and to the minimal vertex separators of the graph.

%%%%%%%%%%%%%%%%%%%%%%%%%%%%%%%%%%%%%%
%
%%%%%%%%%%%%%%%%%%%%%%%%%%%%%%%%%%%%%%

\section{Universal vertices, twin vertices and integer Laplacian eigenvalues}

This section is devoted to review known results from the literature concerning 
certain integer Laplacian eigenvalues based on classical invariants such as universal vertices and false and true twins. 
The goal is to show  that such results as rewritten here allow us to determine  these values 
 by means of simple algorithms.

 The proof of the next theorem derives from 
 Corollary 13.1.4 \cite{GR04} 
  which relates the universal vertices of $G$ with the connected components of its complementar graph $\overline G$.
It also derives from  Theorem 4.1.8 \cite{Mo12} that states that if a graph $G$ is a connected graph of order $n$
  then $n$ is a Laplacian eigenvalue of $G$ if and only if $G$ is a join of  two graphs.

\begin{theorem}\label{theo:universal}
Let $G$ be a connected non-complete graph of order $n$. 
If $G$ has $k$ universal vertices, $n$ is a Laplacian eigenvalue of $G$ with multiplicity $k$.
\end{theorem}
\begin{proof} 
Since $G$ can be expressed as the join of a graph induced by the  set of universal vertices and 
the graph induced by the remaining vertices, the theorem mentioned above can be applied.
\qed
\end{proof}

\medskip
Let $G=(V,E)$ be a connected graph and let $F \subseteq V$ ($T \subseteq V$) be
a set of false twins (true twins) of order $k$.
We can observe that the graph induced by $F$ is an  independent set of size $k$.

The  next  result can be  found in \cite{Ab12}  using the concept of clusters.
We present here a different proof.

\begin{theorem} \label{theo:falsetwins}
Let $G$ be a connected graph with a set of false twins $F $ such that each vertex has degree $d$. 
Then $d$ is an integer Laplacian eigenvalue of $L(G)$ with multiplicity at least $|F| -1$.
\end{theorem}

\begin{proof}
Consider $L(G)$ the Laplacian matrix of $G$ with the set of false twins $F$ labeled $v_1, \ldots, v_k$. 
Label their neighbors consecutively  by $v_{k+1}$,...,$v_{d+k}$ 
and the remaining vertices of the graph by $v_{d+k+1}$,...,$v_n$. 
If $v_i, v_j \in F$ then  $(v_i, v_j) \not \in E$ and $d_i = d_j = d$.
So, the Laplacian matrix of $G$ can be written as

$$ 
L(G)=\left [ 
\begin{array}{r|c|c}
d I_{k \times k}    &  -J_{k \times d}   &  0_{k \times (n-k-d)} \\ \hline
-J_{d \times k} &   \star    & \star    \\
0_{n-k-d \times k}   &    \star   & \star \\
\end{array} \right ],
$$

where undefined elements $(\star)$ substitute the  integer elements of the $(n-k) \times(n-k)$-submatrix of $L(G)$ 
after taking out the $k$ first lines and columns of $L(G)$. 
 For $2 \leq i \leq k, e_1- e_i$ are eigenvectors of $L(G)$ associated to $d$ 
where $e_i$ is the $i-th$ vector of the canonical basis of $\mathbb{R}^n$. 
So, $d$ is an integer Laplacian eigenvalue of $L(G)$ with multiplicity at least $k-1 = |F|-1$.
\qed
\end{proof}

The following theorem is a redraft of a result due to Grone and Merris \cite{GM94}. 

\begin{theorem}\label{theo:truetwins}
Let $G$ be a connected graph with a set of true twins  $T$  such that each vertex has degree $d$. 
Then $d+1$ is an integer Laplacian eigenvalue of $L(G)$ with multiplicity at least $|T| -1$.
\end{theorem}

\begin{proof}
The prove is similar to that given in Theorem \ref{theo:falsetwins} where it is enough to change the block $d I_{k \times k}$
 in the Laplacian matrix $L(G)$ by the block  $(d +1)I_{k \times k}$.
 \qed
\end{proof}

Figure \ref{fig:false-true} ilustrates the computation of integer Laplacian eigenvalues of a graph by Theorems \ref{theo:falsetwins} and \ref{theo:truetwins}.
  $F=\{g,h,i\}$ is a set of false twins;
each vertex has degree $2$ and, from
Theorem \ref{theo:falsetwins}, $G$ has $2$ as a Laplacian eigenvalue with multiplicity 2. 
$T=\{k, l, j\}$ is a set of true twins; each vertex has degree $4$. 
By Theorem \ref{theo:truetwins}, $G$ has $5$ 
as a Laplacian eigenvalue with multiplicity $2$.

\begin{figure}[!h]
\begin{center}
\begin{tikzpicture}
  [scale=.32,auto=left]
\tikzstyle{every node}=[circle, draw, fill=white,
                         inner sep=0pt, minimum width=15pt]

%grafo G
\node (a) at (7,-3){$a$};
\node (b) at (4,0){$b$};
\node (f) at (10,0){$f$};
\node (d) at (7,7) {$d$};
\node (c) at (4,4){$c$};
\node (e) at (10,4){$e$};
\node (g) at (12,9){$g$};
\node (i) at (15,7){$i$};
\node(h) at (14,4){$h$};
\node(j) at (0,0){$j$};
\node(k) at (0,4){$k$};
\node(l) at (-2,2){$l$};
\foreach \from/\to in {a/b, c/d,e/f, f/a, 
g/d,g/e,i/d,i/e,h/d,h/e,
j/b, j/c, j/k, j/l, k/b,k/c,k/l, l/b,l/c} 
\draw (\from) -- (\to);
\node at (18,-3) [draw=none,fill=none] {$G$};
\end{tikzpicture}
\end{center}
\caption{Laplacian eigenvalues $\times$ false and true twins}
\label{fig:false-true} 
\end{figure}
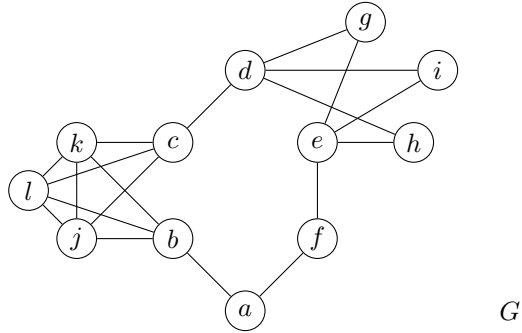

\subsection{Algorithmic aspects}

If  $G=(V,E)$ has universal vertices, false or true twin sets,
it is possible to efficiently recognize the existence
of some of its integer Laplacian eigenvalues, 
by means of simple algorithms.

\begin{itemize}
\item Consider as input  the neighborhood $N(v)$ and the closed neighborhood $N[v]$, 
represented by ordered lists, for all $v \in V$.

\item Apply a lexicographic ordering in both sets of lists;
all false twins and true twins will appear together in the resulting lists.
Using the radix sort \cite{W90}, the result can be found in $O(n^2)$ time complexity.
\end{itemize}

Observe that  in order to determine the universal vertices,
it is  necessary to  test  if $|N(v)| = n-1$.
In this case, the time complexity is $O(n+m)$.

%%%%%%%%%%%%%%%%%%%%%%%%%%%%%%%%%%%
%%
%%%%%%%%%%%%%%%%%%%%%%%%%%%%%%%%%%%

\section{Results on strictly chordal graphs}

The  {\em block duplicate graphs} were introduced by Golumbic and Peled \cite{GP02};
it is a graph obtained by adding zero or more true twins
to the vertices of a block graph. 
The class, a subclass of chordal graphs, was also  defined as {\em strictly chordal graphs}  
by \cite{K05}  and
it  was proved to be {\em gem-free} and  {\em dart-free} \cite{GP02, KLY06} (see Figure \ref{fig:gem-dart}).

In this section we present new results for this class.
In Subsection 3.1, an important 
result due to Cardoso and Rojo  \cite{CR17} is rewritten in Theorem \ref{theo:CardosoRojo}
 taking into account  only integer Laplacian eigenvalues. 
In Subsection 3.2, properties of strictly chordal graphs are stated. 
Based on them, in the last subsection, we prove two new results.
Theorem \ref{theo:ibiconect} shows how Theorem \ref{theo:truetwins} can be applied directly
 to the minimal vertex separators of 
strictly chordal graphs and  Theorem \ref{theo:simplicialB(S)}, our main result, gives  the minimum
number of integer Laplacian eigenvalues of a strictly chordal graph through particular sets of simplicial vertices.

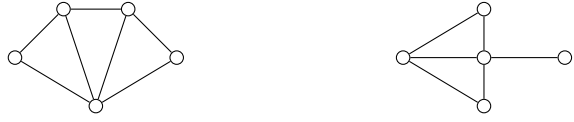
\begin{figure}[h]
\begin{center}
\begin{tikzpicture}
 [scale=.43,auto=left]
 \tikzstyle{every node}=[circle, draw, fill=white,
                        inner sep=1.8pt, minimum width=4pt]

 % [scale=.4,auto=left,every node/.style={circle,fill=black}]
  \node (c) at (1,10) {};
  \node (d) at (3,10) {};
  \node (b) at (-0.5,8.5)  {};
  \node (e) at (4.5,8.5)  {};
  \node (a) at (2,7) {};

  \foreach \from/\to in {c/d,c/b,c/a,d/a,d/e,b/a,e/a}
    \draw (\from) -- (\to);

  \node (n) at (14,10) {};
  \node (k) at (11.5,8.5) {};
 \node (l) at (14,8.5)  {};
  \node (m) at (16.5,8.5)  {};
  \node (j) at (14,7)  {};
  \foreach \from/\to in {n/k,n/l, k/l, l/m,k/j, l/j} 
    \draw (\from) -- (\to);
 \end{tikzpicture}
\end{center}
\caption{Gem and  dart  graphs} 
\label{fig:gem-dart}
\end{figure}

%%%%%%%%%%%%%%%%%%%%%%%

\subsection{Clusters and some integer Laplacian eigenvalues}

 Let $G$ be a connected chordal graph of order $n$ with a set of $k$ false twins $F$.  
 Let $S$ be the set of $\ell$ neighbors of each vertex of $F$. 
The pair   
 $(F,S)$ is denoted a {\em $(k,\ell)$-cluster}  or, simply, a {\em cluster} of $G$. 
 It is immediate that $S$ is a minimal vertex separator of the graph. 
 Cardoso and Rojo \cite{CR17} have defined clusters for any connected graph using the concept of co-neighbor vertices. 
 Also they built  a new graph $H$ from graph $G$ by insertion of edges between pairs of false twins vertices and 
 proved that $SpecL(G) \cap SpecL(H)$ is not empty. 
  
 Let $G$ be a connected graph of order $n$ having a $(k,\ell)$-cluster $(F,S)$. 
 Let $H$ be a graph of order $k$. 
 Then $G(H)$ denotes the graph obtained from $G$ when the vertices of $H$ are identified with the vertices of $F$ 
and it is simply the original graph $G$  
 by adding edges between one or more pairs of vertices of $F$. 
  The clusters $(F_1,S_1)$ and $(F_2,S_2)$ are disjoint in $G$ if $F_1\cap F_2 = \emptyset$ and $S_1 \cap S_2 =  \emptyset$. 
 In the more general form, Cardoso and Rojo  \cite{CR17} gave the following definition. 
  Consider $G$ having $t$ pairwise disjoint clusters $(F_1,S_1), (F_2,S_2),...,(F_t,S_t).$ For $1 \leq j \leq t$, 
let $H_j$ be a graph of order $|F_j|$. 
Then $G(H_1, . . . , H_t)$ denotes the graph obtained from $G$ where the vertices of each graph $H_j$ 
is  identified with the vertices in $F_j$. 
It follows that $V(H_j) = F_j$, $V(G(H_1,..., H_t)) = V (G)$ and 
$E(G(H_1, . . . , H_t)) = E(G) \cup  E(H_1)\cup  . . . \cup E(H_t).$
 We denote $F = \bigcup_{j=1,...t} F_j$ and $FS =\bigcup_{j=1,...t} (F_j \cup S_j)$.
 
 The next result shows that the Laplacian eigenvalues of $G(H_1,...,H_t)$ remain the same, 
 independently of the graphs $H_1,....,H_t$, with the exception of $|F_1|+....+|F_t| - t$ of them. 
 Let $\widetilde{L}(G-F)$ be the principal submatrix of $L(G)$ obtained after deleting the rows and columns 
 of $L(G)$ with indices in $F$. 

\begin{theorem}\label{theo:CardosoRojo} {\rm \cite{CR17}} 
Let $G$ be a connected graph of order $n$ with $t \geq 1$ pairwise disjoint clusters $(F_1,S_1),...,(F_t,S_t)$. 
For  $j = 1,...,t$, assume that $|F_j| = k_j$, $|S_j| = \ell_j$ and 
 each graph $H_j$ is defined as above to obtain $G(H_1, . . . , H_t)$
such that  
$L_j \textbf{1}_{k_j} = \mu _{k_j } \textbf{1}_{k_j}$ where $L_j$ is the Laplacian matrix of $H_j$.  It follows that

\begin{equation}\label{eq:1}
det(\lambda I - L(G(H_1,...,H_t))) = p_L(\lambda) \prod_{j =1,t} \prod_{i =1,k_j -1} (\lambda - (\ell_j + \mu_i(L_j)))
\end{equation}
where $p_L (\lambda)$ is the characteristic polynomial of the matrix $\widetilde{L} (G-FS)$ 
whose degree is $n - \sum_{j=1,.., t} k_j + t$. In particular, when the graphs $H_1,...,H_t$ 
are the empty graphs, $G(H_1,...,H_t) = G$ and
\begin{equation}\label{eq:2}
det(\lambda I - L(G)) = p_L(\lambda )\prod_{j=1,t} (\lambda - \ell_j)^{k_j-1}.
\end{equation}
\end{theorem}

%%%%%%%%%%%%%%%%%%

\subsection{Chordal and strictly chordal graphs}

A chordal graph is a graph in which every cycle of length four and greater has a cycle chord. 
Basic concepts about  chordal graphs are assumed to be known and 
can be found in Blair and Peyton \cite{BP93}  and Golumbic \cite{Go04}.
Following, the most pertinent concepts are reviewed.

A subset $S \subset V$ is a {\em separator} of $G$ if at least two vertices in 
the same connected component of $G$ are in two distinct connected components of
$G[V\setminus S]$. 
The set $S$ is a {\em minimal separator} of $G$ if $S$ is a
separator and no proper set of $S$ separates the graph.
The set of minimal separators of $G$ is denoted by ${\mathbf S}$.

Let $G = (V, E)$ be a chordal graph and $u,v  \in V$. 
A subset $S \subset V$  is a {\em vertex separator}  for
non-adjacent vertices $u$  and $v$  (a $uv$-separator) if the
removal of $S$ from the graph separates $u$ and $v$  into distinct
connected components. 
If no proper subset of $S$  is a $uv$-separator then $S$ is a {\em minimal $uv$-separator}. 
When the pair of vertices remains unspecified, we refer to $S$  as a {\em
minimal vertex separator} ({\em mvs}). 
The set of minimal vertex separators  of $G$ is denoted by $\mathbb{S}$.
A linear algorithm to determine the set of  minimal vertex separators can be find in \cite{MP10}. 
A minimal separator of $G$ is always a minimal vertex separator
but the converse is not true.

A {\em strictly chordal graph } is  obtained by adding zero or more true twins to each vertex of a block graph $G$.
A  new characterization based on minimal vertex separators was presented by Markenzon and Waga \cite{MW15}.
Based on the characterization theorem, a recognition algorithm  becomes very simple.

\begin{theorem}\label{theo:caract2}
Let  $G=(V,E)$ be a chordal graph and $\mathbb S$ be the set of minimal vertex separators of $G$.
The following statements are equivalent: 
\begin{enumerate}
  \item $G$ is a strictly chordal graph.
  \item For any distinct $S, S^{\prime} \in {\mathbb S}$, $S \cap S^{\prime}= \emptyset$.
  \item $G$ is gem-free and dart-free. 
\end{enumerate}
\end{theorem}

Interesting properties of strictly chordal graphs can be stated. 

\begin{property}\label{prop:separator-mvs}{\rm\cite{M20}}
Let $G$ be a strictly chordal graph, $\mathbb S$ the set of minimal vertex separators of $G$ and 
$\mathbf S$ the set of minimal separators of $G$.
Then $\mathbb S= \mathbf S$.
\end{property}
\begin{proof}
In \cite{M20}, this property was proved for the non-inclusion chordal graphs. 
As strictly chordal graphs are a subclass of non-inclusion graphs, 
the  result follows.
\qed
\end{proof}

\begin{property}\label{prop:strictly-boundary}
Let $G$ be a strictly chordal graph.
Then all boundary cliques contain exactly one minimal vertex separator.
\end{property}
\begin{proof}
By definition, a maximal clique $Q$ is called a boundary clique if there exists a maximal clique $Q^\prime$
such that  $Q \cap Q^\prime $ is the set of non-simplicial vertices of $Q$. 
The set  $Q \cap Q^\prime $ is a clique; so it is a minimal vertex separator. 
As $G$ is a strictly chordal graph  there is not proper containment of separators, 
then a boundary clique contains exactly one minimal vertex separator.
\qed
\end{proof}

\begin{property}\label{prop:strictly-true}
Let $G$ be a strictly chordal graph and $ S$ a minimal vertex separator of $G$.
Then all vertices of $S$ are true twins.
\end{property}
\begin{proof}
By the definition,  a strictly chordal graph is obtained by adding zero or more true twins
to the vertices of a block graph. 
The separators of a block graph have cardinality one.
Adding true twins to this vertex  results in a separator of greater cardinality with the same set of neighboors. \qed
\end{proof}

%%%%%%%%%%%%%%%%%%%%%%%%%

\subsection{New results}

The next result shows how Theorem \ref{theo:truetwins} can be applied directly to the minimal vertex
separators of strictly chordal graphs.

\begin{theorem}\label{theo:ibiconect}
Let $G=(V,E)$ be a strictly chordal graph and  let
 $\mathbb{S }$ be the set of minimal vertex separators of $G$.
 Let $S \in \mathbb S$ and $v$ a vertex belonging to $S$.
 Then $d(v) + 1$  is an integer Laplacian eigenvalue of $L(G)$  with multiplicity at least $|S|-1$.
\end{theorem}
\begin{proof}
As $G$ is a strictly chordal graph, by Theorem \ref{theo:caract2}, 
the minimal vertex separators are two by two disjoint sets.
For each $S$, all vertices are true twins (Property \ref{prop:strictly-true}).
By Theorem \ref{theo:truetwins}, the result follows.
\qed
\end{proof}

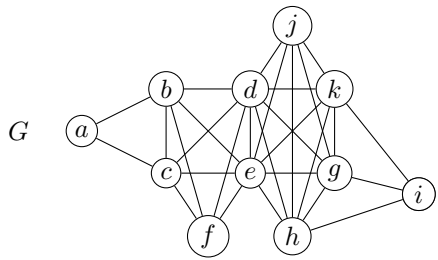
\begin{figure}[!h]
\begin{center}
\begin{tikzpicture}
  [scale=.28,auto=left]
  \tikzstyle{every node}=[circle, draw, fill=white,
                        inner sep=1.8pt, minimum width=4pt]

\node (a) at (-4,-2){$a$};
\node (b) at (0,0){$b$};
\node (d) at (4,0){$d$};
\node (c) at (0,-4){$c$};
\node (e) at (4,-4){$e$};

\node (j) at (6,3){$j$};
\node (k) at (8,0){$k$};
\node (g) at (8,-4){$g$};
\node (h) at (6,-7){$h$};
\node (i) at (12,-5){$i$};
\node (f) at (2,-7){$\small f$};
\node (i) at (12,-5){$i$};

\node at (-7,-2) [draw=none,fill=none] {$G$};

\foreach \from/\to in {a/b, a/c,  b/c, b/e, b/d,b/f,c/d,c/e,c/f,d/e,d/f, d/k,d/g,d/h,e/g,e/f,e/k,e/h,k/g, k/h,k/i,g/i,g/h,h/i, j/d, j/e, j/h, j/g, j/k}
 \draw (\from) -- (\to);
 \end{tikzpicture}
 \end{center}
\caption{Laplacian eigenvalues $\times $ separators in a strictly chordal graph} 
\label{fig:theo6}
\end{figure}

\medskip
Observe that, if $G$ is a biconnected strictly chordal graph, there is at least one non zero integer eigenvalue 
of $L(G)$ for each minimal vertex separator of $G$.  

\medskip
Figure \ref{fig:theo6}  ilustrates how Theorem 6 can easily determine 
$6, 9$ and $7^{(2)}$ as integer Laplacian eigenvalues of the graph which spectrum is:

$$SpecL(G) = [0; 1.18541; 2.61293; 3.72314; 5.64590; 6; 6.55734; 7^{(2)}; 9; 9.27527].$$

The next result, Theorem \ref{theo:simplicialB(S)}, gives  the minimum number of 
integer Laplacian eigenvalues of a strictly chordal graph that derives from the  quantity of some simplicial vertices of the graph. 

Let $\mathbb Q$ be the set of maximal cliques.
For each minimal vertex separator $S$, $S \in \mathbb S$, let us denote by $B(S)$
the set of boundary cliques that contain $S$.
Each maximal clique belonging to $B(S)$ can be partioned in two
subsets: $S$ and $P$, a set of simplicial vertices.
If $|B(S)| > 1$ and each maximal clique has exactly one simplicial vertex,
$B(S)$ is a {\em cluster}, as defined by Cardoso and Rojo \cite{CR17}. 

\begin{theorem} \label{theo:simplicialB(S)}
Let $G$ be a strictly chordal graph and ${\mathbb S}^*$ the set of minimal vertex separators
of cardinality $t$, 
 such that, for each $S_i \in  {\mathbb S}^*$, $|B(S_i)| >1$. 
 Let ${\mathcal P}_i$ be the set of simplicial vertices belonging to  $B(S_i)$.
Then $G$ has at least $\sum_{i=1,...,t} \left( |{\mathcal P}_i|-1 \right)$ integer Laplacian eigenvalues.
\end{theorem}
\begin{proof}
Let $G$ be a strictly chordal graph and  $S_i \in \mathbb{S}^{*}$ a minimal vertex separator of $G$.
Let $|S_{i}|=s_i$, $i= 1, \ldots, t$, $|B(S_i)| = b_i$ and $Q_{ij} \in B(S_i)$,  $j=1,...,b_i$. 
Let $P_{ij}$ be the set of simplicial vertices of  $Q_{ij}$, where $Q_{ij} = P_{ij} \cup S_i$.
As $G$ is a chordal graph, the subgraph induced
 by $P_{ij}$ is the complete graph $K_{n_{ij}}$ of order $n_{ij}$. 
 Moreover, the subgraph induced by ${\mathcal P}_i = \cup_{j=1,..,b_i} P_{ij}$ is 
 $H_i= \cup_{j=1,...,b_i} K_{n_{ij}}$. 
 So, $|{\mathcal P}_i| =\sum_{j=1,...,b_i} n_{ij}$.  

Since $G$ is a strictly chordal graph, their minimal vertex separators are two by two disjoint sets.
From Theorem $4$ (\ref{eq:1}), for each $i, i=1,...,t$, 
 all Laplacian eigenvalues of each $H_i$ are $ s_i + n_{ij}, j=1,...,b_i$. 
Since $G$ has $t$ induced subgraphs $H_i$ with $|{\mathcal P}_i|$ simplicial vertices,
 $G$ has at least $\sum_{i=1,...,t} \left( |{\mathcal P}_i|-1\right)$ integer Laplacian eigenvalues.
\qed
\end{proof}

\medskip
As already mentioned, if $G$ is a strictly chordal graph, their minimal vertex separators are two by two disjoint sets. 
In order to simplify the notation in the remaining of the text, we work with 
 only one minimal vertex separator $S \in \mathbb{S}^{*}$.

\begin{corollary}\label{cor:integer-values}
Let $S \in \mathbb S^*$, $B(S) = \{Q_1, \ldots, Q_{b}\}$  and $P_k$ the set of simplicial vertices of $Q_k$.
Then the following values are some of the integer Laplacian eigenvalues of $G$:
\begin{enumerate}
\item $|Q_k|$ with multiplicity $|P_k| -1$, $\forall Q_k \in B(S)$;
\item $|S|$ with multiplicity $b-1$.
\end{enumerate}
\end{corollary}

\begin{corollary}\label{cor:number}
Let $S \in \mathbb S^*$, $B(S) = \{Q_1, \ldots, Q_{b}\}$  and $P_k$ the set of simplicial vertices of $Q_k$.
Let $\mathcal P$ be the set of  simplicial vertices in $B(S)$ and $\mathcal F$, the set of false twins in $\mathcal P$.
The number of integer Laplacian eigenvalues uniquely provided by Theorem \ref{theo:simplicialB(S)} is:

 $$\left\{ \begin{array} {l l }
|{\mathcal P}|-1 - \sum_{k=1,b} (|P_k| -1) & \textrm{if  } {\mathcal F} = \emptyset; \\
|{\mathcal P}| - \sum_{k=1,b} (|P_k| -1) - |{\mathcal F}| & \textrm{otherwise.}
\end{array} \right.$$
\end{corollary}

\bigskip

The following example ilustrates these results.
Consider the graph $G$ of Figure \ref{ex4}.  
There is a minimal vertex separator $S=\{d\}$ of  $G$, such that 
$B(S) = \{\{d,e,f\}, \{d,g,h\},\{d,i,j\},\{d,k,l,m\}\}$. 
So, ${\mathcal P}=\{e,f,g,h,i,j,k,l,m\}$. 
Since  $|B(S)|> 1$, by Corollary \ref{cor:integer-values}, 
there are $ 8$ integer Laplacian eigenvalues of   $G$:  $4^{(2)}, 3^{(3)}, 1^{(3)}$.
Some of them, $4^{(2)}, 3^{(3)}$, are already provided by Theorem \ref{theo:truetwins}.   
The remaining three integer eigenvalues are uniquely provided by Theorem \ref{theo:simplicialB(S)};
see the spectrum bellow.

$$ SpecL(G) = [0; 0.23941; 1^{(3)}; 1.53342; 3^{(3)}; 3.21582; 4^{(2)}; 11.01135 ].$$

\begin{figure}[!h]
\begin{center}
\begin{tikzpicture}
  [scale=.28,auto=left]
  \tikzstyle{every node}=[circle, draw, fill=white,
                        inner sep=1.8pt, minimum width=4pt]

\node (j) at (16,-2){$a$};
\node (k) at (20,-2){$b$};
\node (l) at (24,-2){$c$};
\node (m) at (28,-2){$d$};
\node (n1) at (31,5){$e$};
\node (n2) at (32,3){$f$};
\node (n3) at (33,1){$g$};
\node (n4) at (34,-1){$h$};
\node (n5) at (34,-3){$i$};
\node (n6) at (33,-5){$j$};
\node (n7) at (32,-7){$k$};
\node (n8) at (31,-9){$l$};
\node (n9) at (28,-7){$m$};

\node at (10,-2) [draw=none,fill=none] {$G$};

\foreach \from/\to in {j/k,k/l,l/m,m/n1,m/n2,m/n3,m/n4,m/n5,m/n6,m/n7,m/n8,n1/n2,n3/n4,n5/n6,n7/n8, n8/n9,n7/n9, m/n9} 
    \draw (\from) -- (\to);
 \end{tikzpicture}
 \end{center}
\caption{Laplacian eigenvalues $\times$ boundary cliques}
\label{ex4} 
\end{figure}
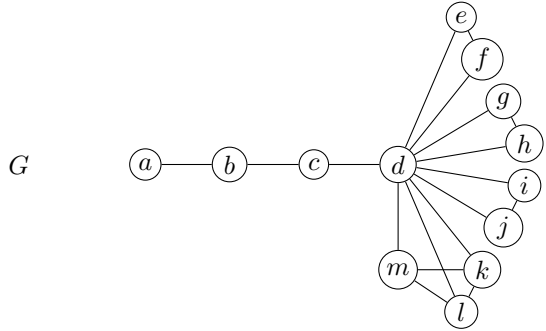

A strictly chordal graph can have more simplicial vertices than the ones described in 
Theorem \ref{theo:simplicialB(S)}.
Let $Q \in {\mathbb Q}$ and let $P$ be the set of simplicial vertices belonging
to $Q$. If $Q$ does not belong to any $B(S)$ being $S \in {\mathbb S}^*$ then,
by Theorem \ref{theo:truetwins}, we can conclude that the value $|Q|$ with
multiplicity $|P| - 1$ is an integer Laplacian eigenvalue of $G$.

%%%%%%%%%%%%%%%%%%%%%%%%%%%

\subsubsection{Algorithmic aspects}

For the class of strictly chordal graphs, several integer Laplacian eigenvalues are directly deduced from its
structural properties, as was seen in 
Theorems \ref{theo:truetwins}, \ref{theo:ibiconect},  \ref{theo:simplicialB(S)} and Corollary \ref{cor:integer-values}.
It is interesting to highlight the fact that  the determination of these eigenvalues has linear time complexity.
Given  $G=(V,E)$, the following steps are needed:

Step 1:  determine the set of maximal cliques $\mathbb Q$  and the set of minimal vertex separators $\mathbb S$ of $G$.
By Theorem \ref{theo:caract2}, all minimal vertex separators are pairwise disjoint, so each vertex of $G$
can be labeled as a simplicial vertex or as belonging to exactly one minimal vertex separator $S$.
This step can be accomplished in linear time complexity \cite{MP10}.

Step 2:  test if $|S| >1$, for each $S\in \mathbb S$;
in this case, by Theorem \ref{theo:ibiconect}, $d(v)+1$, for $v \in S$, with multiplicity $|S|-1$ is a Laplacian eigenvalue.
As each vertex belongs to at most one  minimal vertex separator, $\sum_{S\in \mathbb S}  |S| < n$.
This step has linear time complexity.

Step 3: determine $B(S)$ for each $S\in \mathbb S$. 
By Property \ref{prop:strictly-boundary}, all boundary cliques contain exactly one minimal vertex separator.
So, a sequential search through the maximal cliques is enough to determine $B(S)$
for all minimal vertex separators.
This step must search through the vertices of all maximal cliques; 
  the time complexity of this step is $O(n+m)$.

Step 4: For each $S \in \mathbb S$, test if $|B(S)| >1$.
If yes,  compute, for  $B(S)= \{Q_{1}, \ldots Q_{b}\}$,  the set $P_i$ of simplicial vertices of $Q_i$.
By Corollary \ref{cor:integer-values}, there are the following integer Lapalcian eigenvalues in $G$:
\begin{enumerate}
\item $|Q_i|$ with multiplicity $|P_i| -1$, $\forall Q_i \in B(S)$;
\item $|S|$ with multiplicity $b-1$.
\end{enumerate}
If no and $|B(S)| = 1$, then $B(S)= \{Q\}$ and $P$ is the set of simplicial vertices of $Q$.
Thus  $|Q|$  with  multiplicity $|P| - 1$  is an integer Laplacian eigenvalue of $G$.
This step must search all maximal cliques belonging to $B(S)$ for each $S\in \mathbb S$. 
As  $\cup_{S \in {\mathbb S}} B(S) \subseteq {\mathbb Q}$  the time complexity of this step is $O(n+m)$.

\medskip 
Step 5: determine the  maximal cliques $Q$ that do not belong to any $B(S)$.
For each one of these cliques, determine the set of simplicial vertices, $P$.
By Theorem 3, $|Q|$,  with  multiplicity $|P| - 1$,  is an integer Laplacian eigenvalue of $G$.
As this step must also search all maximal cliques of the graph then its time complexity is $O(n+m)$.

\section{Conclusions}

In this paper, some important results already known in the literature were reviewed in terms of
structural parameters of the graph.
Theorems \ref{theo:universal}, \ref{theo:falsetwins} and  \ref{theo:truetwins}
 provide some integer Laplacian eigenvalues of a connected graph,  
 enabling us to  efficiently implement their computation.
For the class of strictly chordal graphs, some new integer Laplacian eigenvalues are revealed, which can be seen in 
Theorem 7.
Also interesting to highlight is the fact that  the determination of a minimum number of integer Laplacian eigenvalues 
of a strictly chordal graph has linear time complexity,
provided by resourceful graph algorithms.

\section*{Acknowledgments}

This work was supported by grants 304177/2013-0 and 304706/2017-5, CNPq, Brazil.

%*******************************************************************************
%	Bibliography								*
%*******************************************************************************


\begin{thebibliography}{10}\label{bibliography}

\bibitem{Ab12} 
{\small 
{\sc N.M.M. Abreu, D.M. Cardoso, E.A. Martins, M. Robbiano, B. San Mart\'\i n:}
{\it On the Laplacian and signless Laplacian spectrum of a graph with k pairwise co-neighbor vertices.} 
Linear Algebra Appl., {\bf 437} (2012), 2308--2316.}

\bibitem{Ab18} 
{\small
{\sc N.M.M. Abreu, C.M. Justel, L. Markenzon, C.S. Oliveira, C.F.E.M. Waga:} 
{\it Block-indifference graphs: Characterization, structural and spectral properties.}
Discrete Appl. Math. {\bf 269} (2019) 60--67.}

\bibitem{Ab19} 
{\small
{\sc N.M.M. Abreu, C.M. Justel, L. Markenzon:}
{\it Integer Laplacian Eigenvalues of Chordal Graphs.}
 Linear Algebra Appl.,  available online 28 December 2019.}

\bibitem{BP93} 
{\small
{\sc J.R.S. Blair, B. Peyton:}
{\it An Introduction to Chordal Graphs and Clique Trees.}
In Graph Theory and Sparse Matrix Computation, IMA {\bf 56} (1993) 1--29. }

\bibitem{CR17} 
{\small
{\sc D. Cardoso, O. Rojo:}
{\it Edge perturbation on graphs with clusters: 
Adjacency, Laplacian and signless Laplacian eigenvalues.}
 Linear Algebra Appl., {\bf 512} (2017) 113--128.}

\bibitem{EB17} 
{\small
{\sc E. Estrada, M. Benzi:} 
{\it Core-satellite graphs: Clustering, assortativity and spectral properties.}
Linear Algebra Appl., {\bf 517} (2017) 30--52.}


\bibitem{FKMN05} 
{\small
{\sc S.M. Fallat, S. J. Kirkland, J. J. Molitierno,  M. Neumann:}
{ \it On Graphs Whose Laplacian Matrices Have Distinct Integer Eigenvalues.} 
Graph Theory, {\bf 50} (2005) 162--174.}

\bibitem{FI95} 
{\small
{\sc I. Faria:} 
{\it Multiplicity of Integer Roots of Polynomials of Graphs.} 
Linear Algebra Appl., {\bf 229} (1995) 15--35.}

\bibitem{Go04} 
{\small
{\sc M.C. Golumbic:} 
{\it Algorithmic Graph Theory and Perfect Graphs.} 
$2^{nd}$ edition, Academic Press, New York, 2004.}

\bibitem{GP02}
{\small
{\sc M.C. Golumbic, U.N.  Peled:}
{\it Block duplicate graphs and a hierarchy of chordal graphs.}
 Discrete Appl. Math., {\bf 124} (2002) 67--71.  }

\bibitem{GR04} 
{\small
{\sc C. Godsil, G.Royle:} 
{\it Algebraic Graph Theory.} 
Springer, New York, 2004.}


\bibitem{GMS90}
{\small
{\sc R. Grone, R. Merris,  V.S. Sunder:}
{\it The Laplacian Spectrum of a graph.} 
SIAM J. Matrix Anal. Appl., {\bf 2} (1990) 218--238.  }

\bibitem{GM94} 
{\small
{\sc R. Grone, R. Merris:}
 {\it The Laplacian Spectrum of a graph II.}  
SIAM J. Discrete Math. {\bf 7} (1994) 221--229.  }

\bibitem{Ha63} 
{\small
{\sc F. Harary:}  
{\it A characterization of block-graphs.} 
Canad. Math. Bull.  {\bf 6} (1963) 1--6.  }


\bibitem{K05}
{\small
{\sc W. Kennedy:}
{ \it Strictly chordal graphs and phylogenetic roots.}
Master Thesis, University of Alberta, 2005. }

\bibitem{KLY06}
{\small
{\sc W. Kennedy, G. Lin, G. Yan:}
{ \it Strictly chordal graphs are leaf powers.} 
J. Discrete Algorithms, {\bf 4} (2006) 511--525.  }


\bibitem{Ki05} 
{\small
{\sc S. Kirkland:}
{ \it Completion of Laplacian integral graphs via edge addition.}
 Discrete Math. {\bf 295} (2005) 75--90.}

\bibitem{K10} 
{\small
{\sc S. Kirkland, M.A.A. Freitas, R.R. Del Vecchio, N.M.M. Abreu:} 
{\it Split non-threshold Laplacian integral graphs.} 
Linear Algebra Appl. {\bf 58} (2010) 221--233.   }

\bibitem{M20} 
{\small
{\sc L. Markenzon:}
{\it Non-inclusion and other subclasses of chordal graphs.}
Discrete Appl. Math. {\bf 272} (2020) 43--47.  }

\bibitem{MP10} 
{\small
{\sc L. Markenzon, P.R.C.  Pereira:}
{\it One phase algorithm for the determination of minimal vertex separators of chordal graphs.}
Int. Trans. Oper. Res. {\bf 17} (2010) 683--690.   }


\bibitem{MW15} 
{\small
{\sc L. Markenzon,  C.F.E.M. Waga:} 
{\it New results on ptolemaic graphs.}  
Discrete Appl. Math., {\bf 196} (2015) 135--140. }

\bibitem{Me94} 
{\small
{\sc R. Merris:}
{ \it Degree Maximal Graphs Are Laplacian Integral.} 
Linear Algebra Appl., {\bf 199}  (1994) 381--389. }


\bibitem{Mo12} 
{\small
{\sc J. Molitierno:}
{\it Applications of Combinatorial Matrix Theory to Laplacian Matrices of Graphs.}
CRC Press, 2012.   }

\bibitem{W90} 
{\small
{\sc M.A. Weiss:} 
{\it Data Structures and Algorithms.}
$2^{nd}$ edition, Benjamin-Cummings Publishing Co., 
 Redwood City, CA, US, 1993.  }



\end{thebibliography}
\end{document}